%% file: main.tex
\documentclass[11pt]{article}
\usepackage[page,toc,titletoc,title]{appendix}
\usepackage[utf8]{inputenc}
\usepackage{float}
\input{myheader.tex}
\usepackage[a4paper, margin=1.0in]{geometry}
\usepackage{bbm}
\usepackage{todonotes}
\DeclareMathOperator{\cl}{cl}

\def\Disc{{\mathrm{Disc}}}         

\DeclareMathOperator{\cc}{CC}        
\DeclareMathOperator{\DD}{D}
\DeclareMathOperator{\RR}{R}

\def\1{\mathbf{1}}                       

\date{}

\begin{document} 

\title{A counter-example to the probabilistic universal graph conjecture via randomized communication complexity}

\author{
Lianna Hambardzumyan \thanks{School of Computer Science, McGill University. \texttt{lianna.hambardzumyan@mail.mcgill.ca}} 
\and Hamed Hatami \thanks{School of Computer Science, McGill University. \texttt{hatami@cs.mcgill.ca}. Supported by an NSERC grant.}
\and Pooya Hatami \thanks{Department of Computer Science and Engineering, The Ohio State University. \texttt{pooyahat@gmail.com}. Supported by NSF grant CCF-1947546}}

\maketitle

\begin{abstract}
We refute the Probabilistic Universal Graph Conjecture of Harms, Wild, and Zamaraev, which states that a hereditary graph property admits a constant-size probabilistic universal graph if and only if it is stable and has at most factorial speed. 

Our counter-example follows from the existence of a sequence of $n \times n$ Boolean matrices $M_n$, such that their public-coin randomized communication complexity tends to infinity, while the randomized communication complexity of every $\sqrt{n}\times \sqrt{n}$ submatrix of $M_n$ is bounded by a universal constant. 
\end{abstract}

\section{Introduction}
 The field of communication complexity studies the amount of communication required to solve the problem of computing discrete functions when the input is split between two or more parties. In the most commonly studied framework, there are two parties, often called Alice and Bob, and a communication problem is defined by a Boolean matrix $M=[m_{ij}]_{n \times n}$, where   \emph{Boolean} means that the entries are either $0$ or $1$.   Alice receives a row number, and Bob receives a column number $j$. Together, they should both compute the entry $m_{ij}$ by exchanging bits of information in turn, according to a previously agreed-on protocol. There is no restriction on their computational power; the only measure we care to minimize is the number of exchanged bits.

A deterministic protocol $\pi$ specifies how the communication proceeds. It specifies what bit a player sends at each step. This bit depends on the input of the player and the history of the communication so far. It is often assumed that the last communicated bit must be the output of the protocol.   A protocol naturally corresponds to a binary tree as follows. Every internal node of the tree is associated with either Alice or Bob.  If an  internal node $v$ is associated with  Alice, then it is labeled with a Boolean function $a_v$, which prescribes the bit sent by Alice at this node as a function of her input $i$. Similarly, the nodes associated with Bob are labeled with Boolean functions of $j$.  Each leaf is labeled by $0$ or $1$ which corresponds to the output of the protocol at that leaf.

We denote the number of bits exchanged on the input $(i,j)$ by $\cost_\pi(i,j)$, which is precisely the length of the path from the root to the corresponding leaf. The \emph{communication cost} of the protocol is simply the depth of the protocol tree, which is the maximum of $\cost_\pi(i,j)$  over all inputs $(i,j)$. 
$$\cc(\pi) \coloneqq \max_{i,j}  \cost_\pi(i,j).$$

 We say that $\pi$ computes $M$  if $\pi(i,j)=m_{ij}$ for all $(i,j)$, where $\pi(i,j)$ denotes the protocol's output on the input $(i,j)$.   The \emph{deterministic communication complexity} of $M$, denoted by $\DD(M)$, is the smallest communication cost of a protocol that computes $M$.  It is easy to see that $\DD(M)\le \ceil{\log(n)}+1$ as Alice can use $\ceil{\log(n)}$ bits to send her entire input to Bob, and Bob knowing the values of both $i$ and $j$, can send back $m_{ij}$. 

It is well-known that if the deterministic communication complexity of a  matrix is bounded by a constant $c$, then the matrix is highly structured -- its rank is at most $2^c$, and it can be partitioned into at most $2^c$  all-zero and all-one submatrices~\cite{MR1426129}. These facts characterize the family of matrices that satisfy $\DD(M)=O(1)$.  A fundamental problem in communication complexity, with connections to harmonic analysis and operator theory~\cite{Hat21}, is to obtain a  characterization of families of matrices that have $O(1)$ \emph{randomized} communication complexity. 

A (public-coin) \emph{randomized protocol} $\pi_R$ of cost $c$ is simply a probability distribution  over the deterministic protocols of cost $c$. Given an input $(i,j)$, to compute $m_{ij}$, Alice and Bob use their shared randomness to sample a deterministic protocol from this distribution, and execute it. 

We say that the error probability of $\pi_R$ is at most $\epsilon$ if $\Pr[\pi_R(i,j) \neq m_{ij}] \le \epsilon$ for every input $(i,j)$.  For $\epsilon \in (0,1/2)$, let $\RR_\epsilon(M)$ denote the smallest cost of a randomized protocol that computes $M$ with error probability at most $\epsilon$. Note that $\epsilon=1/2$ can be easily achieved by  outputting a random bit; hence it is crucial that $\epsilon$ is defined to be strictly less than $1/2$. It is common to take $\epsilon=\frac{1}{3}$. Indeed, the choice of $\epsilon$ is not important as long as $\epsilon \in (0,1/2)$, since the probability of   error can be reduced to any  constant $\epsilon'>0$ by repeating the same protocol  independently for some $O(1)$ times, and outputting the most frequent output.   We denote $\RR(M) \coloneqq \RR_{1/3}(M).$

It is well-known that the $n \times n$ identity matrix $\I_n$ satisfies $\RR(\I_n)\le 3$ and $\DD(\I_n) =\ceil{\log(n)}+1$.  Hence, in contrast to the deterministic case, there are matrices with $\RR(M)=O(1)$ that have arbitrarily large rank. 

There are very few known examples of matrix classes that have  randomized communication complexity $O(1)$~\cite{Hat21,harms2021randomized,harms2020universal}. Let $\cM=(M_n)_{n \in \mathbb{N}}$ be a sequence of $n \times n$ Boolean matrices $M_n$, and define $\RR(\cM): n \mapsto \RR(M_n)$.   Let us look at some necessary conditions for $\cM$ to satisfy $\RR(\cM)=O(1)$. 

Let $\cl(\cM)$ denote  the \emph{closure} of $\cM$, defined as the set of all square matrices that are a submatrix of some $M_n$. Note that $\cl(\cM)$ is the smallest such hereditary property that contains all the matrices in $\cM$, where a set of square matrices is called \emph{hereditary} if it is closed under taking square submatrices.

Let $\GT_k$ denote the $k \times k$ \emph{Greater-Than}   matrix defined as $\GT_k(i,j)=1$ if and only if $i \le j$. The sequence $\cM$ is called \emph{stable}, if there exists $k\in \mathbb{N}$ such that   $\GT_k \not\in \cl(\cM)$. It is well-known~\cite{MR3439794,RamamoorthyS15} that $\RR(\GT_k)= \Omega(\log \log k)$ which tends to infinity as $k$ grows. Hence, if $\RR(\cM)=O(1)$, then $\cM$ must be stable. The term stability is coined due to Shelah's unstable formula theorem in model theory which characterizes unstable theories by the nonexistence of countably infinite half-graphs~\cite{MR1083551}, where half-graphs are the graphs with biadjacency matrix $\GT_k$ for some $k$. Stable graph families are known to have useful properties such as strong regularity lemmas~\cite{malliaris2014regularity} and the Erd\H{o}s-Hajnal property~\cite{chernikov2018note}. 

The second necessary condition for $\RR(\cM)=O(1)$ follows from a bound on the number of matrices with $O(1)$ randomized communication complexity.  A standard derandomization argument, \cref{prop:CountR}, shows that the number of such $n\times n$ matrices  is bounded by $2^{O(n \log n)}$. Consequently, if $\RR(\cM)=O(1)$, then  $|\cl(\cM)_n|\leq 2^{O(n \log n)}$, where  $\cl(\cM)_n$ denotes the set of $n \times n$ matrices in $\cl(\cM)$. Thus, in the terminology of graph theory~\cite{MR1490438,MR1769217}, the speed of growth of $|\cl(\cM)_n|$ is at most \emph{factorial}. 

Connections between randomized communication complexity and implicit graph representations were discovered in \cite{harms2020universal, harms2021randomized}. Inspired by the \emph{Implicit Graph Conjecture}~\cite{MR1186827, MR1971502} and its connection to the growth rate of hereditary graph properties,   Harms, Wild and Zamaraev~\cite{harms2021randomized} formulated a probabilistic version of the  Implicit Graph Conjecture, which translates to the following statement in communication complexity (See \cite[Conjecture 1.2 and Proposition 1.6]{harms2021randomized}). We refer the readers to \cite{harms2021randomized} for an in-depth discussion of connections between communication complexity and implicit representations of graphs. 

\begin{conjecture}[Probabilistic Universal Graph Conjecture~\cite{harms2021randomized}]
\label{conj:PUG}
Let $\cM$  be a sequence of $n \times n$ Boolean matrices. Then $\RR(\cM)=O(1)$ if and only if $\cM$ is stable and  $|\cl(\cM)_n|\leq 2^{O(n \log n)}$. 
\end{conjecture}

\cref{conj:PUG}  speculates that the  two necessary conditions for $\RR(\cM)=O(1)$ that we discussed above are also sufficient. In other words, they  characterize Boolean matrices  that have randomized communication complexity $O(1)$. It is shown in \cite{harms2021randomized} that \cref{conj:PUG} is true for matrix sequences corresponding to restricted classes of hereditary graph families such as monogenic bipartite families, interval graphs, and families of bounded twin-width.

In this article, we prove the following theorem which refutes~\cref{conj:PUG}. 

\begin{theorem}[Main Theorem]
\label{thm:main}
There exists a stable sequence $\cM$ of  Boolean matrices $(M_n)_{n \in \mathbb{N}}$ such that $\RR(M_n)=\Theta(\log(n))$ and $|\cl(\cM)_n|\leq 2^{O(n \log n)}$. 
\end{theorem}
Note that every $n\times n$ matrix $M$ satisfies $\RR(M)= O(\log n)$. In particular, the above construction shows that this maximum is achievable even for stable hereditary matrix families of speed at most factorial.

Furthermore, as a consequence of \cref{conj:PUG},   \cite{harms2021randomized} speculates that the randomized communication complexity of every hereditary property of Boolean matrices $\cM$ with at most factorial speed has a gap behavior, either $\RR(\cM)=O(1)$ or $\RR(\cM)=\Omega(\log \log n)$. We refute this weaker conjecture as well. In particular, \cref{main:thm2}, proved in \cref{sec:mainthm2}, shows that for every growing function $w(n)<10^{-3}\log n$, there exists a matrix sequence $\cM=(M_n)_{n\in \N}$ such that $\RR(M_n)=w(n)$, and every $\sqrt{n}\times \sqrt{n}$ submatrix $F$ of $M_n$ satisfies $\RR(F)=O(1)$. As the proof of \cref{thm:main} demonstrates, if we take $w(n)$ to be any function that is $\omega(1)$ and $o(\log \log (n))$, then  $\cl(\cM)$ is a hereditary matrix property with factorial speed and $\RR(\cM)=\Theta(w(n))$. 

We present the proof of \cref{thm:main}, which builds on \cref{main:thm2}, in \cref{sec:thmmain}. 

\paragraph{Related work.} Shortly after this paper was written, in a follow-up paper, \cite{hatami2021implicit} refuted the Implicit Graph Conjecture. The proof of~\cite{hatami2021implicit} combines the ideas of this paper with an information-theoretic counting argument. We  remark that even though \cref{conj:PUG} was inspired by the Implicit Graph Conjecture, the refutation of the latter in \cite{hatami2021implicit} does not imply the results of the present paper.


\paragraph{Acknowledgement} We wish to thank Zachary Hunter for pointing to us a minor error in an earlier draft of this article. 

\section{Preliminaries}

All logarithms in this article are in base $2$. For a positive integer $n$, we denote $[n]=\{1,\ldots,n\}$.  We use the standard Bachmann-Landau asymptotic notations: $O(\cdot)$, $\Omega(\cdot)$, $\Theta(\cdot)$, $o(\cdot)$, and $\omega(\cdot)$. 

The Cartesian product $A \times B$ of two sets $A,B \subseteq [n]$ is called a \emph{combinatorial rectangle}. We will need the following lower bound on randomized communication complexity. 

\begin{definition}
Let $M$ be an $n \times n$ Boolean matrix, and let $\mu$ be a probability distribution on $[n] \times [n]$.  The discrepancy of a combinatorial rectangle $R  \subseteq [n] \times [n]$  under  $\mu$ is defined as
$$ \Disc_{\mu}(M,R) = \left|\Pr_{\mu}[m_{ij}=1 \text{ and } (i,j) \in R] - \Pr_{\mu}[m_{ij}=0 \text{ and } (i,j) \in R]\right|.$$
The discrepancy of $M$ under $\mu$ is defined as 
$\Disc_{\mu}(f) = \max_{R}\{\Disc_{\mu}(M,R)\}$,
where the maximum is over all combinatorial rectangles  $R$.
\end{definition}

\begin{theorem}\cite[Proposition 3.28]{MR1426129}
\label{thm:DiscLower}
Let $M$ be an $n \times n$ Boolean matrix, and let $\mu$ be a probability distribution on $[n] \times [n]$. Then for every $\epsilon>0$, 
$$\RR_{\frac{1}{2}-\epsilon}(M) \ge  \log \frac{2\epsilon}{ \Disc_\mu(M)}. $$
In particular, 
\begin{equation}
\label{eq:discLow}
\RR(M) \ge  \log \frac{1}{3 \Disc_\mu(M)}. 
\end{equation}
\end{theorem}

As discussed in the introduction, stability is a necessary condition for a matrix sequence to satisfy $\RR(\cM)=O(1)$. The next proposition proves a second necessary condition: an upper bound on $|\cl(\cM)_n|$. 

\begin{proposition}
\label{prop:CountR}
The number of $n\times n$ matrices $M$ with $\RR(M) \le c$ is  $2^{O(2^cn \log n)}$. 
\end{proposition}
\begin{proof}
Let $M$ be an $n\times n$ Boolean matrix with $\RR(M) \leq c$. For every such $M$, there is a distribution $\mu_M$ over deterministic protocols $\pi$ of cost $c$ such that 
$$
\Pr_{\pi\sim \mu_M}[M(i,j)=\pi(i,j)]\geq \frac{2}{3} \qquad \text{for all $i,j$.}
$$
By the Chernoff bound, the error probability of the protocol can be reduced to strictly less than $\frac{1}{n^2}$ by sampling $O(\log n)$ independent samples from $\mu_M$ and outputting the majority outcome. Thus by union bound, there exists $t=O(\log n)$ deterministic protocols $\pi_1,\ldots, \pi_t$, each of cost $c$, such that for every $i$ and $j$, 
\begin{equation}\label{eq:derandcounting}
    M(i,j)= \mathsf{majority}\{\pi_1(i,j), \ldots, \pi_t(i,j)\}. 
\end{equation}

Next, we show that the number of deterministic protocols of cost $c$ is at most $2^{O(2^c n)}$. Every such protocol corresponds to a binary tree of depth at most $c$, which has $O(2^c)$ nodes. Every node is associated with one of the two players, and the communicated bit is determined by the input of the corresponding player according to a function $[n] \to \{0,1\}$. Thus there are $2^{n+1}$ possible choices for each node of the tree. Overall, this bounds the number of such protocols by $2^{O(2^cn)}$. 

Finally, since every matrix $M$ can be described in the form of \cref{eq:derandcounting}, and there are $2^{O(2^cn)}$ choices for each $\pi_i$, the number of such matrices is at most $2^{O(2^c n \log n)}$. 
\end{proof}

\section{Proof of \cref{thm:main}}\label{sec:thmmain}
The proof will rely on the following theorem, which involves a probabilistic argument presented in \cref{sec:mainthm2}.
\begin{theorem}
\label{main:thm2}
Let $w:\mathbb{N}\to \mathbb{N}$ be a non-decreasing function satisfying $ w(n) \to \infty$ and $w(n)\le 10^{-3} \log(n)$. For every sufficiently large $n$, there exists an $n \times n$ Boolean matrix $M$ with the following properties.  
\begin{enumerate}[label=(\roman*)]
    \item $ \RR(M)  = w(n).$
    \item  Every $\sqrt{n} \times \sqrt{n}$ submatrix $F$ of $M$  satisfies    $\RR(F) =O(1).$ 
\end{enumerate}
\end{theorem}

Let $w(n)= \floor{10^{-3}\log(n)}$, and for every sufficiently large $n$, let $M_n$ be the matrix that is  guaranteed to exist by \cref{main:thm2}. For smaller values of $n$, let $M_n$ be an arbitrary $n \times n$ Boolean matrix and let $\cM$ denote the corresponding sequence. By \cref{main:thm2}~(i), we have $\RR(\cM)=\Theta(\log(n))$, and by \cref{main:thm2}~(ii), $\cM$ is stable. 

It remains to bound $|\cl(\cM)|_n$. Let $F$ be an $n\times n$ matrix in $\cl(\cM)$. There are two cases to consider:

\begin{enumerate}
    \item $F$ is a submatrix of an $M_k$ for $k > n^2$. In this case, by \cref{main:thm2}~(ii), $\RR(F) =O(1)$. So by~\cref{prop:CountR}, the number of such matrices is bounded by $2^{O(n \log n)}$.
    \item $F$ is a submatrix of an $M_k$ with $n \le k \le n^2$. The number of such matrices is at most 
    $$n^2 {n^2 \choose n}^2 = 2^{O(n \log n)}.$$
\end{enumerate}

We conclude that the total number of $n\times n$ matrices in $\cl(\cM)$ is $ 2^{O(n \log n)}$ as desired. 

\section{Proof of \cref{main:thm2}}\label{sec:mainthm2}
We will use a probabilistic argument to show the existence of an $n \times n$ matrix $M$ that satisfies $\RR(M) \ge w(n)$, and the property (ii).  Note that modifying a row of a matrix can change its randomized communication complexity by at most $1$. Hence, to guarantee $\RR(M) = w(n)$, we can replace the rows of $M$ to all-zero rows, one by one, until we achieve $\RR(M) = w(n)$. We will also show that for our construction, (ii) will remain valid under such modifications. 

Let $M=[m_{ij}]_{n \times n}$ be selected uniformly at random from the set of all Boolean $n \times n$ matrices that have  exactly $rn$ number of $1$'s where $r=2^{3w(n)} \leq n^{0.01}$. Denote $p=\frac{rn}{n^2}=\frac{r}{n}$ to be the fraction of $1$'s.

Similarly, define $M'=[m'_{ij}]_{n\times n}$ be the Boolean matrix with independent entries, where each entry $m'_{ij}=1$ with probability $p$.  

\paragraph{Lower-bounding $\RR(M)$:}
Let $\epsilon=\frac{1}{3 \cdot 2^{w(n)}}$. We will show that with high probability, the discrepancy of $M$ is bounded by $\epsilon$, and thus $\RR(M)$ is large. 

Let $M_0=\{(i,j) \ | \ m_{ij}=0\}$ and  $M_1=\{(i,j) \ | \ m_{ij}=1\}$, and define $M'_0$ and $M'_1$ similarly from the entries of $M'$. Define the probability distribution $\mu$ on $[n] \times [n]$   as 
	$$\mu:(i,j) \mapsto \left\{ 
	\begin{array}{lcr}
	\frac{1}{2rn} & \qquad & m_{ij}=1 \\
	\frac{1}{2(n-r)n}& &  m_{ij}=0 \\
	\end{array}
	\right. $$
Note that $\mu$ is defined so that it assigns a total measure of $\frac{1}{2}$ uniformly to each of $M_0$ and $M_1$. 

Consider a fixed $a\times b$ combinatorial rectangle $R\subseteq [n]\times [n]$. Then, 
\begin{eqnarray*}
\Disc_\mu(R)&=& \left|\frac{|R \cap M_1|}{2rn}- \frac{|R \cap M_0|}{2(n-r)n}\right|= \left|\frac{|R \cap M_1|}{2rn}- \frac{|R|-|R \cap M_1|}{2(n-r)n}\right|  \\
&=&\left|\frac{n|R \cap M_1|- r|R|}{2nr(n-r)}\right|= \left|\frac{|R \cap M_1|- p|R|}{2r(n-r)}\right|\le \left|\frac{|R \cap M_1|- abp}{rn}\right|.
\end{eqnarray*}

Therefore, 
\begin{equation}\label{eqn:discr}
\Pr\left[\Disc_\mu(R) \ge \epsilon \right] \le \Pr\left[\left||R\cap M_1|- abp \right| \geq \epsilon rn\right]
\end{equation}

Note that
\begin{align*}
\Pr[|R\cap M_1|\geq abp + \epsilon r n] &\leq \Pr\left[|R\cap M'_1|\geq abp+\epsilon rn \ \big\vert \ |M'_1|\geq rn\right] \\ &\leq  \frac{\Pr\left[|R\cap M'_1|\geq abp+\epsilon rn\right]}{\Pr[|M'_1|\geq rn]} \\ &\leq \exp \left(-\Omega\left(\frac{(\frac{\epsilon rn}{abp})^2\cdot abp}{2+\frac{\epsilon rn}{abp}}\right) \right) = 2^{-\Omega(\epsilon^2 r n )}, 
\end{align*}
where the last inequality is the Chernoff bound $\Pr[|R\cap M'_1|-\mu \geq \delta \mu]\leq e^{\frac{-\delta^2\mu}{2+\delta}}$ applied with parameters $\delta= \epsilon r n/abp$ and $\mu=abp$. Similarly,
\begin{align*}
\Pr[|R\cap M_1|\leq abp - \epsilon r n] &\leq \Pr\left[|R\cap M'_1|\leq abp-\epsilon rn \ \big\vert \ |M'_1|\leq rn\right] \\ &\leq  \frac{\Pr\left[|R\cap M'_1|\leq abp+\epsilon rn\right]}{\Pr[|M'_1|\leq rn]} \\ &\leq  2^{-\Omega(\epsilon^2 r n )},
\end{align*}
where we applied the Chernoff bound $\Pr[|R\cap M'_1|-\mu \le \delta \mu]\leq e^{\frac{-\delta^2\mu}{2}}$ with the same parameters as above. Taking the union bound over the $2^{2n}$ rectangles $R$, and combining it with \cref{eqn:discr} and the bounds above, we obtain
$$
\Pr\left[\Disc_\mu(R) \ge \epsilon \right] \leq 2^{2n}2^{-\Omega(\epsilon^2 r n)}= o(1), 
$$
where we used the fact that $\epsilon = \frac{1}{3\cdot 2^{w(n)}}$, $r=2^{3w(n)}$, and $w(n)\rightarrow \infty$. Namely, $n= o(\epsilon^2 rn)$. 
Finally, applying the discrepancy lower bound of \cref{eq:discLow}, we obtain that for sufficiently large $n$, 
$$\Pr\left[\RR(M) \le w(n)\right] =  \Pr\left[\RR(M) \le 
\log\frac{1}{3\epsilon}\right] \le \Pr[\Disc_\mu(M) \ge \epsilon]  =  o(1). $$

\paragraph{Verifying (ii):}   We prove that with probability $1-o(1)$, for every $a,b\le \sqrt{n}$, every $a \times b$ submatrix of $M$ contains a row or a column with at most four $1$'s. Note that the statement is trivial when $\min(a,b)\le 4$, and hence, we fix  $a,b> 4$. 

If $a \ge b$, then the probability that there is an $a \times b$ submatrix such that each of its $a$ rows contains at least five $1$'s is bounded by
\begin{align*}
{n \choose a} \cdot {n \choose b} \cdot \frac{{b \choose 5}^a {n^2-5a \choose rn-5a}}{{n^2 \choose rn}} &\le n^{a+b}\cdot  \left(\frac{rnb}{n^2}\right)^{5a} \leq n^{2a}  \left(\frac{b}{n^{0.99}}\right)^{5a}\\ &\le n^{(2- 5\times 0.49) a} = o\left(\frac{1}{n}\right).  
\end{align*}
where we used $r=2^{3 w(n)}<n^{0.01}$ and $a>4$.
Similarly, if $a < b$, then the probability that there is an $a \times b$ submatrix such that each of its $b$ columns contains at least five $1$'s is bounded by $o\left(\frac{1}{n}\right)$.

Thus by a union bound over the $n$ choices of $a,b\leq \sqrt{n}$, the probability that there is  $a,b \in [\sqrt{n}]$ and an $a \times b$  submatrix where every column or row contains at least five $1$'s is bounded by $o(1)$.

Now suppose that every $a \times b$ submatrix $F$ of $M$ contains a row or a column with at most four $1$'s. We will show by induction on $a,b$ that in this case, every such $F$ corresponds to the  biadjacency matrix of a disjoint union of four bipartite graphs that are all forests. The base case of $a,b\leq 2$ is trivial. Consider a row (or a column) with at most four $1$'s, and let $e_1, e_2, e_3, e_4$ be the edges corresponding to these (at most) four entries. Removing this row from $F$ will result in a smaller submatrix, which by induction hypothesis, can be written as the union of four forests $\cF_1, \cF_2, \cF_3, \cF_4$. Now $F$ can be decomposed into the union of  four forests $\cF_i\cup  \{e_i\}$ for $i\in [4]$. 

The bound $\RR(F)=O(1)$ follows by first observing that each forest is an edge-disjoint union of two graphs, each   a vertex-disjoint union of stars. Hence, it suffices to show that the biadjacency matrix of any vertex-disjoint union of stars has $O(1)$  randomized communication complexity. Suppose that $G$ is a union of vertex-disjoint stars $S_1,\ldots, S_k$, with a bipartition $(U,V)$. Alice receives $u\in U$ and Bob receives $v\in V$, and they want to decide whether $(u,v)\in E(G)$, which is equivalent to whether $u$ and $v$ belong to the same star. To solve this problem, Alice maps her input $u$ to the index $i$ such that $u \in S_i$. Similarly, Bob maps $v$ to $j$ such that $v \in S_j$. Now they can use the randomized communication protocol for $\I_k$ to check whether $i=j$. This verifies (ii).

Finally, note that if $F$ is a union of four forests, then replacing a row of $F$ with an all-zero row   will not violate this property.  
\bibliographystyle{alpha}
\bibliography{refs}

\end{document}

%% file: myheader.tex
\usepackage{amsmath,amsthm,amssymb,amscd,amstext,amsfonts,mathscinet}
\usepackage{mathtools}
\usepackage{mathrsfs}
\usepackage{thmtools}
\usepackage{latexsym}
\usepackage{verbatim}
\usepackage{framed}
\usepackage{graphicx}
\usepackage{stmaryrd}
\usepackage{enumitem}
\usepackage{fullpage}
\usepackage{bm}
\usepackage{color}
\usepackage{hyperref}
\usepackage{url}
\usepackage{physics}
\usepackage{dsfont}

\usepackage[nameinlink, noabbrev, capitalize]{cleveref}
\hypersetup{colorlinks={true},linkcolor={blue},citecolor=blue}

\theoremstyle{plain}

\newtheorem*{theorem*}{Theorem}

\newtheorem{thm}{Theorem}[section]
\newtheorem{theorem}[thm]{Theorem}

\newtheorem{proposition}[thm]{Proposition}
\newtheorem*{proposition*}{Proposition}
\newtheorem{conjecture}[thm]{Conjecture}
\newtheorem{definition}[thm]{Definition}

\newtheorem*{claim*}{Claim}

\theoremstyle{remark}

\theoremstyle{remark}

\DeclarePairedDelimiter\ceil{\lceil}{\rceil}             
\DeclarePairedDelimiter\floor{\lfloor}{\rfloor}          

\newcommand{\N}{\mathbb{N}}

\newcommand{\cF}{\mathcal F}

\newcommand{\cM}{\mathcal M}

\DeclareMathAlphabet{\mathpzc}{OT1}{pzc}{m}{it}

\def\cost{\mathrm{cost}}     
     
\def\GT{\mathrm{GT}}   
\def\I{\mathtt{I}}

\def\1{\mathbf{1}} 
\def\0{\mathbf{0}}

\let\latexchi\chi
\makeatletter
\renewcommand\chi{\@ifnextchar_\sub@chi\latexchi}
\newcommand{\sub@chi}[2]{
  \@ifnextchar^{\subsup@chi{#2}}{\latexchi^{}_{#2}}%
}
\newcommand{\subsup@chi}[3]{
  \latexchi_{#1}^{#3}%
}
\makeatother